\documentclass{easychair}
\usepackage{amsmath}
\usepackage{amsthm}
\usepackage[noend]{algorithmic}
\usepackage{algorithm}
\usepackage{graphics}
\usepackage{hyperref}

\algsetup{linenosize=\normalsize} % make line numbers full-size
\theoremstyle{plain}
\newtheorem{theorem}{Theorem}
\newtheorem{lemma}{Lemma}

\theoremstyle{definition}
\newtheorem{problem}{Problem}
\newtheorem{definition}{Definition}

\newcommand{\abs}[1]{\left| #1 \right|}
\newcommand{\program}{F} % program to analyze
\newcommand{\inputs}{I} % set of input bits
\newcommand{\branchpoints}{B} % set of branch points
\newcommand{\wt}[1]{\operatorname{wt} (#1)} % weight of a path or input
\newcommand{\branches}[1]{\operatorname{B} (#1)} % branches true along a path
\renewcommand{\path}[1]{\operatorname{P} (#1)} % path triggered by an input
\newcommand{\nonepath}{B_{\mathrm{none}}} % path taking no branches
\newcommand{\nonesupport}{S_{\mathrm{none}}} % variables not in the support of any branch condition
\newcommand{\basispath}[1]{B_{#1}} % basis path taking given branch
\newcommand{\condition}[1]{C_{#1}} % condition of given branch
\newcommand{\support}[1]{S_{#1}} % support variables of given branch
\newcommand{\truecount}[1]{T_{#1}} % number of assignments to the support variables of the given branch which make the branch condition true

\newcommand{\comment}[1]{}
\begin{document}
\title{Speeding Up SMT-Based\\ Quantitative Program Analysis}
\titlerunning{Speeding Up SMT-Based Quantitative Program Analysis}
\author{Daniel J. Fremont \and Sanjit A. Seshia}
\institute{University of California, Berkeley\\\email{dfremont@berkeley.edu\\sseshia@eecs.berkeley.edu}}
\authorrunning{Fremont and Seshia}
\maketitle

\begin{abstract}
Quantitative program analysis involves computing numerical quantities
about individual or collections of program executions. An example of
such a computation is quantitative information flow analysis, where
one estimates the amount of information leaked about secret data
through a program's output channels. Such information can be quantified 
in several ways, including channel capacity and (Shannon) entropy. In
this paper, we formalize a class of quantitative analysis problems
defined over a weighted control flow graph of a loop-free program. These
problems can be solved using a combination of path enumeration, SMT solving, and model counting. However, existing methods can
only handle very small programs, primarily because the number of
execution paths can be exponential in the program size. We show how
path explosion can be mitigated in some practical cases by taking
advantage of special branching structure and by novel algorithm
design. We demonstrate our techniques by computing the channel capacities of the  
timing side-channels of two programs with extremely large numbers of
paths.
\end{abstract}

% keywords: quantitative information flow, channel capacity, model counting
% intro (w/ related work), background & definitions, new techniques, experiments, conclusion

\section{Introduction} % needs to be fleshed out, more motivation added
\label{sec:intro}

Quantitative program analysis involves computing numerical quantities
that are functions of individual or collections of program
executions. Examples of such problems include computing worst-case or
average-case execution time of programs, and
quantitative information flow, which seeks to compute the amount of
information leaked by a program. Much of the work in this area has
focused on extremal quantitative analysis problems --- that is,
problems of finding worst-case (or best-case) bounds on
quantities. However, several problems involve not just finding
extremal bounds but computing functions over multiple (or all)
executions of a program. One such example, in the general area of
quantitative information flow, is to estimate the entropy or channel
capacity of a program's output channel. These quantitative analysis
problems are computationally more challenging, since the number of
executions (for terminating programs) can be very large, possibly
exponentially many in the program size. 

In this paper, we present a formalization and satisfiability modulo
theories (SMT) based solution to a family of quantitative 
analysis questions for deterministic, terminating programs. 
The formalization is covered in detail in
Section~\ref{sec:model}, but we present some basic intuition here.
This family of problems can be defined over
a weighted graph-based model of the program. 
More specifically, considering the
program's control flow graph, one can ascribe weights to nodes or
edges of the graph capturing the quantity of interest (execution time,
number of bits leaked, memory used, etc.) for basic blocks. Then, to obtain the
quantitative measure for a given program path, one sums up the
weights along that path. Furthermore, in order to count the number of
program inputs (and thus executions) corresponding to a program path, one can perform model
counting on the formula encoding the path condition. Finally, to
compute the quantity of interest (such as entropy or channel capacity)
for the overall program, one combines the quantities and model counts
obtained for all program paths using a prescribed formula.

The obvious limitation of the basic approach sketched above is that,
for programs with substantial branching structure, the number of
program paths (and thus, executions) can be exponential in the program size.
We address this problem in the present paper with two ideas.
First, we show how a certain type of ``confluent'' 
branching structure which often occurs in real programs can be
exploited to gain significant performance enhancements. A common
example of this branching structure is the presence of a conditional
statement inside a for-loop, which leads to $2^N$ paths for $N$ loop
iterations. 
In this case, if the branches are proved to be ``independent'' of each
other (by invoking an SMT solver), then one can perform model
counting of individual branch conditions rather than of entire path
conditions, and then cheaply aggregate those model counts. 
Secondly, to compute a quantity such as channel capacity, it is not necessary
to derive the entire distribution of values over all paths.
For this case, we give an efficient algorithm to compute all the values attained by a given quantity (e.g. execution time) over all possible paths ---
i.e., the support of the distribution --- which runs in time polynomial in the sizes of the program and the support.
Our algorithmic methods are particularly tuned to the
analysis of timing side-channels in programs. 
Specifically, we apply our ideas to computing
the channel capacity of timing side-channels for two standard programs which
have far too many paths for previous techniques to handle. 

Our techniques enable the use of SMT methods in a new application, namely quantitative program analyses such as assessing the feasibility of side-channel attacks. While SMT methods are used in other program verification problems with exponentially-large search spaces, na\"ive attempts to use them to compute statistics  like those we consider do not circumvent path explosion. The optimizations that form our primary contributions are essential in making feasible the application of SMT to our domain.

To summarize, the main contributions of this paper include:
\begin{itemize}
\item a method for utilizing special branching structure to reduce the number of model counter invocations needed to compute the distribution of a class of quantitative measures from potentially exponential to linear in the size of the program, and
\item an algorithm which exploits this structure to compute the support of such distributions in time polynomial in the size of the program and the support.
\end{itemize}

The rest of the paper is organized as follows. We present background
material and problem definitions in Sec.~\ref{sec:model}. Algorithms
and theoretical results are presented in Sec.~\ref{sec:algos-theory}.
Experimental results are given in Sec.~\ref{sec:expts} and we conclude
in Sec.~\ref{sec:concl}.

\comment{
consumptions, runtimes, memory requirements, etc. as a result of
allowing execution to proceed along paths with different
statements. Understanding these variations is important in a number of
domains, such as performance characterization or assessment of
vulnerability to side-channel attacks. As an example of the latter,
there has been work in quantitative information flow (QIF) on
automatically finding bounds on how much information an adversary can
learn by observing the overall runtime of a program, where the timing
variation is due either to cache accesses \cite{cacheaudit} or control
flow \cite{KB}. In the case of control flow, these techniques
ultimately require enumerating all execution paths of the
program. Since the number of paths can be exponential in the number of
branch points, these methods are only practical for programs with very
few conditionals. 

In this paper, we propose two abstract problems which subsume these
types of analyses of resource consumption variation due to control
flow. In general path explosion makes the solving of these problems
difficult, but we 
}

\section{Background and Problem Definition}
%\section{Preliminaries}
\label{section-preliminaries}
\label{sec:model}

We present some background material in Sec.~\ref{sec:prelim} and the
formal problem definitions in Sec.~\ref{sec:probdef}.

\subsection{Preliminaries}
\label{sec:prelim}

We assume throughout that we are given a loop-free deterministic program
$\program$ whose input is a set of bits $\inputs$. Our running example for $\program$ will be the standard algorithm for modular exponentiation by repeated squaring, denoted \texttt{modexp}, where the base and modulus are fixed and the input is the exponent. Usually \texttt{modexp} is written with a loop that iterates once for each bit of the exponent. To make \texttt{modexp} loop-free we unroll its loop, yielding for a 2-bit exponent the program shown on the left of Figure \ref{figure-modexp}. Lines \ref{line-modexp-conditional1}--\ref{line-modexp-endloop1} and \ref{line-modexp-conditional2}--\ref{line-modexp-endloop2} correspond to the two iterations of the loop.
\begin{figure}
\centering
\begin{minipage}{1.75in}
\begin{algorithmic}[1]
\STATE $r \leftarrow 1$
\IF{($e$ \& $1$) = $1$} \label{line-modexp-conditional1}
	\STATE $r \leftarrow r b \pmod{m}$
\ENDIF
\STATE $e \leftarrow e >> 1$
\STATE $b \leftarrow b^2 \pmod{m}$ \label{line-modexp-endloop1}
\IF{($e$ \& $1$) = $1$} \label{line-modexp-conditional2}
	\STATE $r \leftarrow r b \pmod{m}$
\ENDIF
\STATE $e \leftarrow e >> 1$
\STATE $b \leftarrow b^2 \pmod{m}$ \label{line-modexp-endloop2}
\RETURN $r$
\end{algorithmic}
\end{minipage}
\begin{minipage}{3in}
\includegraphics[width=3in]{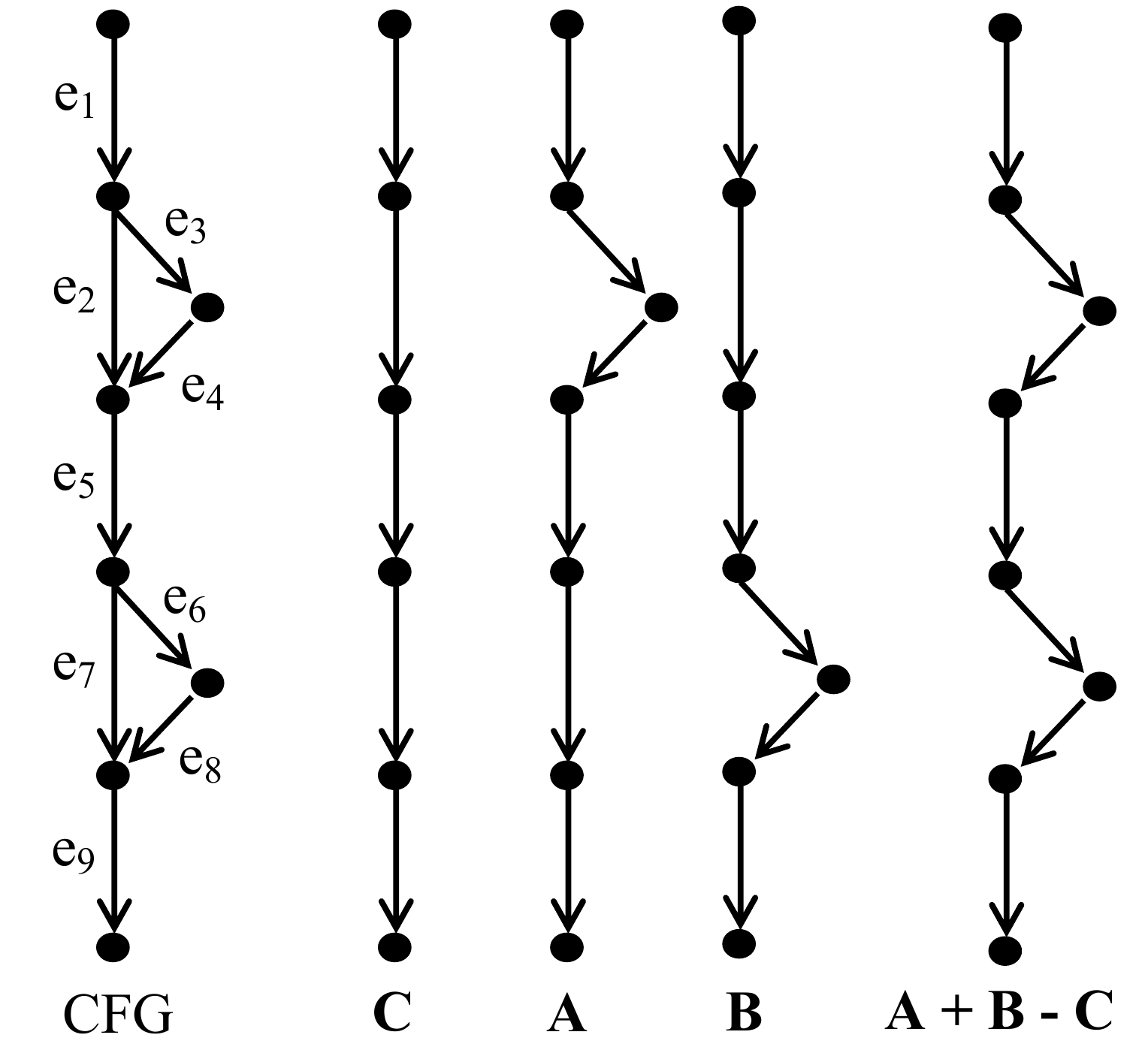}
\end{minipage}
\caption{Unrolled pseudocode and CFG for \texttt{modexp}, computing $b^e \pmod{m}$ for a 2-bit exponent $e$. Paths \textbf{A}, \textbf{B}, and \textbf{C} form a basis, the remaining (rightmost) path being a linear combination of them.}
\label{figure-modexp}
\end{figure}

To describe the execution paths of $\program$ we use the formalism introduced by
McCabe~\cite{basis-paths}. Consider the control-flow graph (CFG) of 
$\program$, where there is a vertex for each basic block, conditionals
having two outgoing edges. For example, since 2-bit \texttt{modexp} has two conditionals, its CFG (shown in Figure \ref{figure-modexp}) has two vertices with outdegree 2. We call such vertices \emph{branch points}, and denote the set of them by $\branchpoints$. Which edge out of a branch point $b \in B$ is taken depends on the truth of its \emph{branch condition} $\condition{b}$, the condition in the corresponding conditional statement. In Figure \ref{figure-modexp}, the branch condition for the first branch point is $(e \& 1) = 1$: if this holds, then edge $e_3$ is taken, and otherwise edge $e_2$ is taken. We model the finite-precision semantics of programs, variables being represented as bitvectors, so that the branch conditions can be expressed as bitvector SMT formulae. Since these conditions can depend on the result of prior computations (e.g. the second branch condition in Figure \ref{figure-modexp}), the corresponding SMT formulae include constraints encoding how those computations proceed. Then each formula uniquely determines the truth of its branch condition given an assignment to the input bits. When necessary, these formulae can be bit-blasted into propositional SAT formulae for further analysis (e.g. model counting).

For convenience we add a dummy vertex to the CFG which
has an incoming edge from all sink vertices. Since $\program$ is
loop-free the CFG is a DAG, and each execution of $\program$
corresponds to a simple path from the source to the (now unique)
sink. Given such a path $P$, we write $\branches{P}$ for the set of branch points where $P$ takes the
right of the two outgoing edges, corresponding to making $\condition{b}$ true. If there are $N$ edges then these paths can be viewed as vectors
in $\{0,1\}^N$, where each coordinate specifies whether the
corresponding edge is taken. For example, in Figure \ref{figure-modexp} path \textbf{A} corresponds to the vector $(1,0,1,1,1,1,0,0,1)$ under the given edge labeling. This representation allows us to speak meaningfully about linear combinations of paths, as long as the result is in $\{0,1\}^N$. A \emph{basis} of the set of paths is defined by analogy to vector spaces to be a minimal set of paths from which all paths can be obtained by taking linear combinations. In Figure \ref{figure-modexp}, the paths \textbf{A}, \textbf{B}, and \textbf{C} form a basis, as the only other path through the CFG can be expressed as $\mathbf{A} + \mathbf{B} - \mathbf{C}$.

Now suppose we are given an integer weight for each basic block of
$\program$, or equivalently for each vertex of its CFG.\footnote{Note that our formalism and approach can be made
to work with rational weights, but we focus here on applications for
which integer weights suffice.}
We define the \emph{total weight} $\wt{P}$ of an execution
path $P$ of $\program$ to be the sum of the weights of all basic
blocks along $P$. Note that we get the same value if the weight of
each vertex is moved to all of its outgoing edges (obviously excluding
the dummy sink), and we sum edge instead of vertex weights --- thus $\wt{\cdot}$ is a linear function.
Since $\program$ is deterministic, each input $x \in \{0,1\}^{\inputs}$ triggers a unique
execution path we denote $\path{x}$, and so has a well-defined total
weight $\wt{x} = \wt{\path{x}}$.

\subsection{Problem Definition}
\label{sec:probdef}

We consider in this paper the following problems:
\begin{problem} \label{problem-distribution}
Picking $x \in \{0,1\}^{\inputs}$ uniformly at random, what
is the distribution of $\wt{x}$? 
\end{problem}
and the special case:
\begin{problem} \label{problem-values}
What is the support of the distribution of $\wt{x}$, i.e. what is the set $\wt{\{0,1\}^{\inputs}} = \{ \wt{x} \: | \: x \in \{0,1\}^{\inputs} \}$?
\end{problem}
One way to think about these problems is to view the weight of a basic
block as some quantity or resource, say execution time or energy, that the block
consumes when executed. Then Problem \ref{problem-distribution} is to
find the distribution of the total execution time or energy
consumption of the program.  

Computing or estimating this distribution is useful in a range of applications
(see~\cite{seshia-acmtecs12}). 
We consider here a quantitative information flow (QIF) setting,
with an adversary who tries to recover $x$ from $\wt{x}$. In the
example above, this would be a timing side-channel attack scenario
where the adversary can only observe the total execution time of the
program. Given the distribution of $\wt{x}$, we can compute any of
the standard QIF metrics such as {\em channel capacity} or {\em
Shannon entropy} measuring how much information is leaked
about $x$. For deterministic programs, the
channel capacity\footnote{Sometimes called the
\emph{conditional min-entropy} of $x$ with respect to $\wt{x}$, since
for deterministic programs with a uniform input distribution they are
the same \cite{qif-foundations}.} is simply the (base 2) logarithm of
the number of possible observed values~\cite{qif-foundations}. Thus to
compute the channel capacity we do not need to know the full
distribution of $\wt{x}$, but only how many distinct values it can
take --- hence our isolation of Problem~\ref{problem-values}. As we
will see, this special case can sometimes be solved much more rapidly
than by computing the full distribution. 

We note that the general problems above can be applied to a variety of
different types of resources. On platforms where the execution time of
a basic block is constant (i.e. not dependent on the state of the
machine), they can be applied to timing analysis. The weights could
also represent the size of memory allocations, or the number of writes
to a stream or device. For all of these, solving Problems
\ref{problem-distribution} and \ref{problem-values} could be useful
for performance characterization and analysis of side-channel
attacks. 

%=============================================================================

%\section{Methods} % need better name for this section! also, should we scrap the subsections?
\section{Algorithms and Theoretical Results}
\label{sec:algos-theory}

The simplest approach to Problem \ref{problem-distribution} would be
to execute program $\program$ on every $x \in \{0,1\}^{\inputs}$, computing
the total weight of the triggered path and eventually obtaining the
entire map $x \mapsto \wt{x}$. This is obviously impractical when
there are more than a few input bits, and is wasteful because often
many inputs trigger the same execution path. A more refined approach
is to enumerate all execution paths, and for each path compute how
many inputs trigger it. This can be done by expressing the branch
conditions corresponding to the path as a bitvector or propositional formula and
applying a \emph{model counter} \cite{gomes-modelcountbookch2009} (this idea was used in \cite{BKR} to count how many inputs led to a given output, although with a linear integer
arithmetic model counter). If the number of paths is much less than
$2^{\abs{\inputs}}$, as is often the case, this approach can be
significantly more efficient than brute-force input
enumeration. However, as noted above the number of paths can be
exponential in the size of $\program$, in which case this approach
requires exponentially-many calls to the model counter and therefore is also
impractical. 

A prototypical example of path explosion is our running example \texttt{modexp}. For
an $N$-bit exponent, there are $N$ conditionals, and all possible
combinations of these branches can be taken, so that there are $2^N$
execution paths. This makes model counting each path infeasible, but
observe that the algorithm's branching structure has two special
properties. First, the conditionals are \emph{unnested}: the two paths
leading from each conditional always converge prior to the next
one. Second, the branch conditions are \emph{independent}: they depend
on different bits of the input. Below we show how we can use these
properties to gain greater efficiency, yielding Algorithms
\ref{algorithm-distribution} and \ref{algorithm-values} for Problems
\ref{problem-distribution} and \ref{problem-values} respectively. 

\subsection{Unnested Conditionals}
\label{sec:unnested}

If $\program$ has no nested conditionals, its CFG has an
``$N$-diamond'' form like that shown in Figure \ref{figure-modexp}
(the number of basic blocks within and between the ``diamonds'' can
vary, of course --- in particular, we do not assume that the ``else'' branch of a conditional is empty, as is the case for \texttt{modexp}). This type of structure naturally arises when
unrolling a loop with a conditional in the body, as indeed is the case
for \texttt{modexp}. Verifying that there are no nested conditionals is
a simple matter of traversing the CFG. 

With unnested conditionals, there is a one-to-one correspondence
between execution paths and subsets of $\branchpoints$, given by $P
\mapsto \branches{P}$. For any $b \in \branchpoints$, we write
$\basispath{b}$ for the path which takes the left edge at every branch
point except $b$ (i.e. makes every branch condition false except for
that of $b$ --- of course it is possible that no input triggers this
path). We write $\nonepath$ for the path which always takes the left edge at each branch point. For example, in Figure \ref{figure-modexp} if the conditionals on lines \ref{line-modexp-conditional1} and \ref{line-modexp-conditional2} correspond to branch points $a$ and $b$ respectively, then $\mathbf{A} = \basispath{a}$, $\mathbf{B} = \basispath{b}$, and $\mathbf{C} = \nonepath$. In general, $\nonepath$ together with the paths $\basispath{b}$ form a basis for the set of all
paths. In fact, for any path $P$ it is easy to see that 
\begin{equation}
\label{equation-path-rep}
P = \left( \sum_{c \in \branches{P}} \basispath{c} \right) - \left( \abs{\branches{P}} - 1 \right) \nonepath \enspace.
\end{equation}
This representation of paths will be useful momentarily.

\subsection{Independence}
\label{sec:independence}

Recall that an input variable of a Boolean function is a \emph{support
variable} if the function actually depends on it, i.e.~the two
cofactors of the function with respect to the variable are not
equivalent. For each branch point $b \in \branchpoints$, let
$\support{b} \subseteq I$ be the set of input bits which are support variables of
$\condition{b}$. We make the following definition: 
\begin{definition}
Two conditionals $b, c \in \branchpoints$ are \emph{independent} if $\support{b} \cap \support{c} = \emptyset$.
\end{definition}
Independence simply means that there are no common support variables,
so that the truth of one condition can be set independently of the
truth of the other. 

To compute the supports of the branch conditions and check
independence, the simplest method is to iterate through all the input
bits, checking for each one whether the cofactors of the branch
condition with respect to it are inequivalent using an SMT query in
the usual way. This can be substantially streamlined by doing a simple
dependency analysis of the branch condition in the source of
$\program$, to determine which input variables are involved in its
computation. Then only input bits which are part of those variables need be
tested (for example, in Figure~\ref{figure-modexp} both branch conditions depend only on the input variable $e$, and if there were other input variables the bits making them up could be ignored). This procedure is outlined as Algorithm
\ref{algorithm-supports}. Note that as indicated in Sec.~\ref{sec:prelim}, the formula $\phi$ computed in line~\ref{supports-smt} encodes the semantics of $\program$ so that the truth of $C_b$ (equivalently, the satisfiability of $\phi$) is uniquely determined by an assignment to the input bits. For lack of space, the proofs of Lemma \ref{lemma-algorithm-supports} and the other lemmas in this section are deferred to the Appendix.

\begin{algorithm}
\caption{FindConditionSupports($\program$)}
\label{algorithm-supports}
\begin{algorithmic}[1]
\STATE Compute CFG of $\program$ and identify branch points $\branchpoints$
\IF{there are nested conditionals}
	\RETURN \texttt{FAILURE}
\ENDIF
\FORALL{$b \in \branchpoints$}
	\STATE $\support{b} \leftarrow \emptyset$ \COMMENT{these are global variables}
	\STATE $\phi \leftarrow$ SMT formula representing $\condition{b}$ \label{supports-smt}
	\STATE $V \leftarrow$ input bits appearing in $\phi$
	\FORALL{$v \in V$}
		\IF{the cofactors of $\condition{b}$ w.r.t. $v$ are not equivalent}
			\STATE $\support{b} \leftarrow \support{b} \cup \{ v \}$
		\ENDIF
	\ENDFOR
\ENDFOR
\RETURN \texttt{SUCCESS}
\end{algorithmic}
\end{algorithm}

\begin{lemma}
\label{lemma-algorithm-supports}
Algorithm \ref{algorithm-supports} computes the supports $\support{b}$ correctly, and given an SMT oracle runs in time polynomial in $\abs{\program}$ and $\abs{\inputs}$.
\end{lemma}

If all of the conditionals of $\program$ are pairwise independent,
then $I$ can be partitioned into the pairwise disjoint sets
$\support{b}$ and the set of remaining bits which we write
$\nonesupport$. For any $b \in \branchpoints$, the truth of
$\condition{b}$ depends only on the variables in $\support{b}$, and we
denote by $\truecount{b}$ the number of assignments to those variables
which make $\condition{b}$ true. Then we have the following formula
for the probability of a path: 
\begin{lemma}
\label{lemma-path-prob}
Picking $i \in \{0,1\}^{\inputs}$ uniformly at random, for any path
$P$, the probability that the path corresponding to input $i$ is $P$
is given by
\begin{equation*}
\Pr \left[ \path{i} = P \right] = \left[ 2^{\abs{\nonesupport}} \left( \prod_{b \in \branches{P}} \truecount{b} \right) \left( \prod_{b \in \branchpoints \setminus \branches{P}} \left( 2^{\abs{\support{b}}} - \truecount{b} \right) \right) \right] / 2^{\abs{\inputs}} \enspace.
\end{equation*}
\end{lemma}

Lemma \ref{lemma-path-prob} allows us to compute the probability of
any path as a simple product if we know the quantities
$\truecount{b}$. Each of these in turn can be computed with a single
call to a model counter, as done in Algorithm
\ref{algorithm-distribution}. 

\begin{algorithm}
\caption{FindWeightDistribution($\program, weights$)}
\label{algorithm-distribution}
\begin{algorithmic}[1]
\IF{FindConditionSupports($\program$) = \texttt{FAILURE} }
	\RETURN \texttt{FAILURE}
\ENDIF
\IF{the sets $\support{b}$ are not pairwise disjoint}
	\RETURN \texttt{FAILURE}
\ENDIF
\FORALL{$b \in \branchpoints$}
	\STATE $\truecount{b} \leftarrow$ model count of $\condition{b}$ over the variables in $\support{b}$
\ENDFOR
\STATE $dist \leftarrow$ constant zero function
\FORALL{execution paths $P$}
	\STATE $p \leftarrow$ probability of $P$ from Lemma \ref{lemma-path-prob}
	\STATE $dist \leftarrow dist [ \wt{P} \mapsto dist(\wt{P}) + p ]$
\ENDFOR
\RETURN $dist$
\end{algorithmic}
\end{algorithm}

\begin{theorem}
Algorithm \ref{algorithm-distribution} correctly solves Problem
\ref{problem-distribution}, and given SMT and model counter oracles
runs in time polynomial in $\abs{\program}$, $\abs{\inputs}$, and the
number of execution paths of $\program$. The model counter is only
queried $\abs{\branchpoints}$ times. 
\end{theorem}
\begin{proof}
Follows from Lemmas \ref{lemma-algorithm-supports} and \ref{lemma-path-prob}.
\end{proof}

Algorithm \ref{algorithm-distribution} improves on path enumeration by
using one invocation of the model counter per branch point, instead of
one invocation per path. In total the algorithm may still take
exponential time, since we need to compute the product of Lemma
\ref{lemma-path-prob} for each path, but if model counting is
expensive there is a substantial savings. 

Further savings are possible if we restrict ourselves to Problem
\ref{problem-values}. For this, we want to compute the possible values of
$\wt{x}$ for all inputs $x$. This is identical to the set of possible
values $\wt{P}$ for all {\em feasible} paths $P$ (the paths
that are executed by some input). Thus, 
we do not need to know the probability
associated with each individual path,
but only which paths are feasible and which are not. Lemma \ref{lemma-path-prob} implies that all paths are feasible (unless some $T_b = 0$ or $T_b = 2^{|S_b|}$, corresponding to a conditional which is identically false or true; then $S_b = \emptyset$, so we can detect and eliminate such trivial conditionals), and this leads to

\begin{lemma}
\label{lemma-submultiset-sums}
Let $D$ be the multiset of differences $\wt{\basispath{b}} -
\wt{\nonepath}$ for $b \in \branchpoints$. Then the possible values of
$\wt{i}$ over all inputs $i \in \{0,1\}^{\inputs}$ are the possible values of
$\wt{\nonepath} + D^+$, where $D^+$ is the set of sums of
submultisets of $D$. 
\end{lemma}

To use Lemma \ref{lemma-submultiset-sums} to solve Problem
\ref{problem-values}, we must find the set $D^+$. The brute-force
approach of enumerating all submultisets is obviously impractical
unless $D$ is very small. We cannot hope to do better than exponential
time in the worst case\footnote{Although we note that for channel
capacity analysis we only need $\abs{D^+}$ and not $D^+$ itself, and
there could be a faster (potentially even polynomial-time) algorithm
to find this value.}, since $D^+$ can be exponentially larger than
$D$. However, in many practical situations $D^+$ is not too much
larger than $D$. This is because the paths $\basispath{b}$ often have
similar weights, so the variation $V = \max D - \min D$ is small and we can apply the following lemma:

\begin{lemma}
\label{lemma-sumset-size}
If $V = \max D - \min D$, then $\abs{D^+} = O(V \abs{D}^2)$. 
\end{lemma}

Small differences between weights are exploited by Algorithm
\ref{algorithm-sums}, which as shown in the Appendix
computes $D^+$ in $O(\abs{D} \abs{D^+})$
time. By Lemma \ref{lemma-sumset-size}, the algorithm's runtime is
$O(|D| \cdot V |D|^2) = O(V \abs{D}^3)$, so it is very efficient when $V$ is small. The
essential idea of the algorithm is to handle one element $x \in D$ at
a time, keeping a list of possible sums found so far sorted so that
updating it with the new sums possible using $x$ is a linear-time
operation. For simplicity we only show how positive $x \in D$ are
handled, but see the analysis in the
Appendix for the general case.  

\begin{algorithm}
\caption{SubmultisetSums($D$)}
\label{algorithm-sums}
\begin{algorithmic}[1]
\STATE $sums \leftarrow (0)$ \label{countsums-init}
\FORALL{$x \in D$} \label{outer-loop}
	\STATE $newSums \leftarrow (sums[0])$ \label{countsums-newsums-init}
	\STATE $i \rightarrow 1$ \COMMENT{index of next element of $sums$ to add to $newSums$}
	\FORALL{$y \in sums$} \label{countsums-inner-loop} % should we indicate here that $sums$ is a list which must be traversed from left to right?
		\STATE $z \leftarrow x + y$ \label{countsums-z}
		\WHILE{$i < \mathsf{len}(sums)$ \AND $sums[i] < z$} \label{countsums-scanloop}
			\STATE $newSums.\mathsf{append}(sums[i])$ \label{countsums-oldsum}
			\STATE $i \leftarrow i + 1$
		\ENDWHILE
		\STATE $newSums.\mathsf{append}(z)$ \label{countsums-newsum}
		\IF{$i < \mathsf{len}(sums)$ \AND $sums[i] = z$} \label{countsums-duplicate}
			\STATE $i \leftarrow i + 1$
		\ENDIF
	\ENDFOR
	\STATE $sums \leftarrow newSums$ \label{countsums-update}
\ENDFOR
\RETURN $sums$ \label{countsums-return}
\end{algorithmic}
\end{algorithm}

Using Algorithm \ref{algorithm-sums} together with Lemma
\ref{lemma-submultiset-sums} gives an efficient algorithm to solve
Problem \ref{problem-values}, outlined as Algorithm
\ref{algorithm-values}. This algorithm has runtime polynomial in the
size of its input and output. 

\begin{algorithm}
\caption{FindPossibleWeights($\program, weights$)}
\label{algorithm-values}
\begin{algorithmic}[1]
\IF{FindConditionSupports($\program$) = \texttt{FAILURE} }
	\RETURN \texttt{FAILURE}
\ENDIF
\IF{the sets $\support{b}$ are not pairwise disjoint}
	\RETURN \texttt{FAILURE}
\ENDIF
\STATE Eliminate branch points with $\support{b} = \emptyset$ (trivial conditionals)
\STATE $D \leftarrow$ empty multiset
\FORALL{$b \in \branchpoints$}
	\STATE $d \leftarrow \wt{\basispath{b}} - \wt{\nonepath}$
	\STATE $D \leftarrow D \cup \{ d \}$
\ENDFOR
\STATE $D^+ \leftarrow \text{SubmultisetSums}(D)$
\RETURN $\wt{\nonepath} + D^+$
\end{algorithmic}
\end{algorithm}

\begin{theorem}
Algorithm \ref{algorithm-values} solves Problem \ref{problem-values}
correctly, and given an SMT oracle runs in time polynomial in
$\abs{\program}$, $\abs{\inputs}$, and
$\abs{\wt{\{0,1\}^{\inputs}}}$. 
\end{theorem}
\begin{proof}
Clear from Lemmas \ref{lemma-algorithm-supports} and
\ref{lemma-submultiset-sums}, and the analysis of Algorithm
\ref{algorithm-sums} (see the Appendix).
\end{proof}

\subsection{More General Program Structure}

As presented above, our algorithms are restricted to loop-free programs which have only unnested, independent conditionals. However, our techniques are still helpful in analyzing a large class of more general programs. Loops with a bounded number of iterations can be unrolled. Unrolling the common program structure consisting of a for-loop with a conditional in the body yields a loop-free program with unnested conditionals. If the conditionals are pairwise independent, as in the \texttt{modexp} example, our methods can be directly applied. If the number of dependent conditionals, say $D$, is nonzero but relatively small, then each of the $2^D$ assignments to these conditionals can be checked for feasibility with an SMT query, and the remaining conditionals can be handled using our algorithms. If many conditionals are dependent then checking all possibilities requires an exponential amount of work, but we can efficiently handle a limited failure of independence. An example where this is the case is the Mersenne Twister example we discuss in Sec.~\ref{sec:expts}, where 2 out of 624 conditionals are dependent. A small level of conditional nesting can be handled in a similar way. In general, when analyzing a program with complex branching structure, our methods can be applied to those regions of the program which satisfy our requirements. Such regions do frequently occur in real-world programs, and thus our techniques are useful in practice.

\section{Experiments}
\label{sec:expts}

% should we try finding a third example where the branch conditions have more than one support bit and so model counting is actually necessary?

% also, a lot of time was wasted (I think) because GameTime recompiles the entire program every time a measurement is made; can we find a way to just compile the program once, and set the initial state through some other method than compiling it in? (I think so, but don't know anything about the simulator and haven't had time to ask Jon about it)

% what tense should this section be written in? 'GameTime was used' or 'GameTime is used' ? I use the present tense in the earlier sections since I think that's the convention, but it sounds strange when talking about our experiments

As mentioned in Sec.~\ref{section-preliminaries}, Problem
\ref{problem-values} subsumes the computation of the channel capacity
of the timing side-channel on a platform where basic blocks have
constant runtimes. To demonstrate the effectiveness of our techniques,
we use them to compute the timing channel capacities of two real-world
programs on the \emph{PTARM} simulator \cite{ptarm}. The tool
\emph{GameTime}~\cite{seshia-tacas11} was used to generate SMT formulae
representing the programs, and to interface with the simulator to
perform the timing measurements of the basis paths. SMT formulae for
testing cofactor equivalence were generated and solved using \emph{Z3}
\cite{z3}. Model counting was done by using Z3 to convert SMT queries
to propositional formulae, which were then given to the model counter
\emph{Cachet} \cite{cachet}. Raw data from our experiments can be obtained at \url{http://math.berkeley.edu/~dfremont/SMT2014Data/}.

The first program tested was the \texttt{modexp} program already
described above, using a 32-bit exponent. With $2^{32}$ paths,
enumerating and model counting all paths is clearly infeasible. Our
new approach was quite fast: finding the branch supports, model
counting\footnote{We note that for this program, each branch condition
had only a single support variable, and thus we have $\truecount{b} =
1$ automatically without needing to do model counting.}, and running
Algorithm \ref{algorithm-sums} took only a few seconds, yielding a timing channel capacity of just over 8 bits. In fact, although the number of paths is very large, the per-path cost of Algorithm \ref{algorithm-distribution} is so low that we were able to compute \texttt{modexp}'s entire timing distribution with it in 23 hours (effectively analyzing more than 50,000 paths per second). The distribution is shown in Figure \ref{figure-modexp-distribution}.

\begin{figure}
\centering
\includegraphics[height=2.1in]{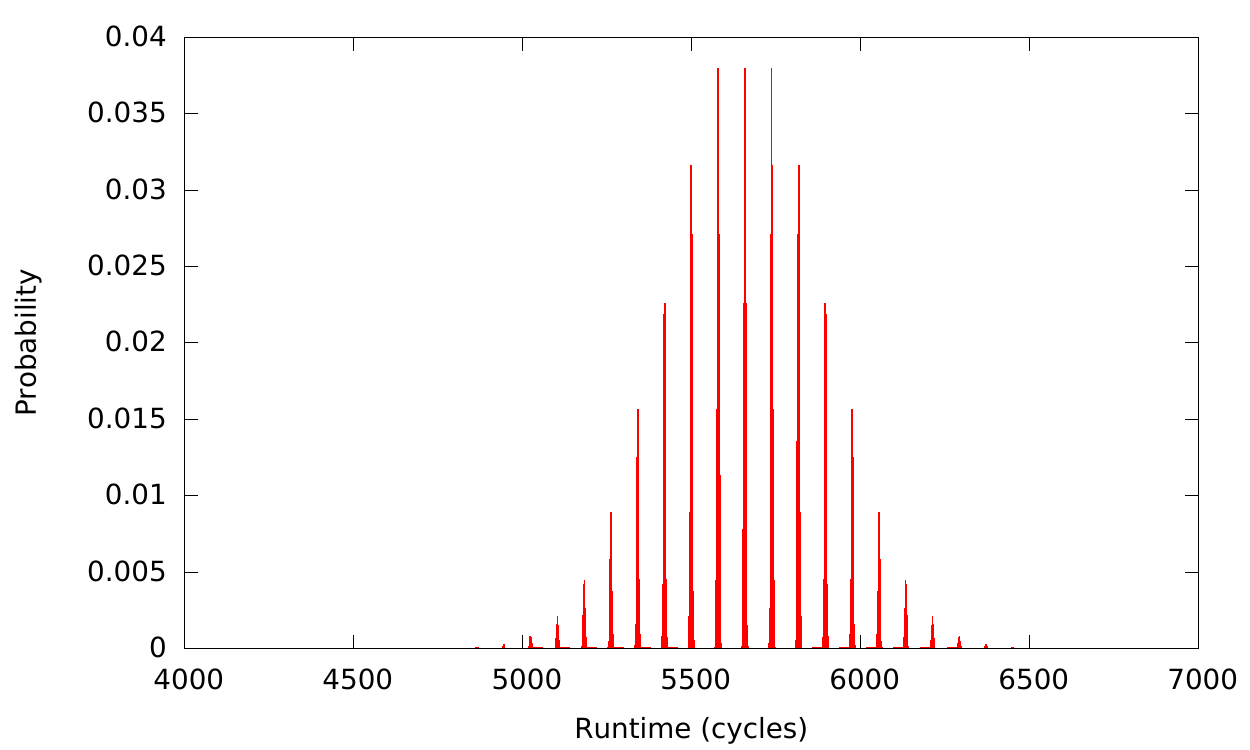}
\caption{Timing distribution of 32-bit \texttt{modexp}, computed with Algorithm \ref{algorithm-distribution}.}
\label{figure-modexp-distribution}
\end{figure}

The second program we tested was the state update function of the
widely-used pseudorandom number generator the Mersenne Twister
\cite{mersenne-twister}. We tested an implementation of the most
common variant, \verb|MT19937|, which is available at
\cite{mt-implementation}. On every 624th query to the generator,
\verb|MT19937| performs a nontrivial updating of its internal state,
an array of 624 32-bit integers. We analyzed the program to see how
much information about this state is leaked by the time needed to do
the update. The relevant portion of the code has $2^{624}$ paths and
thus would be completely impossible to analyze using path
enumeration. With our techniques the analysis became feasible: finding
the branch supports took 54 minutes, while Algorithm \ref{algorithm-sums} took only 0.2 seconds because there was a high level of uniformity across the
path timings. The channel capacity was computed to be around 9.3
bits. We note that among the 624 branch conditions there are two which
are not independent. Thus all four truth assignments to these
conditions needed to be checked for feasibility before applying our
techniques to the remaining 622 conditionals.

%===================================================================
\section{Conclusions}
\label{sec:concl}

We presented a formalization of certain quantitative program analysis
problems that are defined over a weighted control-flow graph
representation. These problems are concerned with understanding how a
quantitative property of a program is distributed over the space of program
paths, and computing metrics over this distribution.
These computations rely on the ability to solve a set of satisfiability
(SAT/SMT) and model counting problems.
Previous work along these
lines has only been applicable to small programs with very few conditionals,
since it typically depends on enumerating all execution paths and the
number of these can be exponential in the size of the program. We
investigated how in certain situations where the number of paths is
indeed exponential, special branching structure can be exploited to
gain efficiency. When the conditionals are unnested and independent,
we showed how the number of expensive model counting calls can be
reduced to be linear in the size of the program, leaving only a very
fast product computation to be done for each path. Furthermore, a
special case of the general problem, which for example is sufficient
for the computation of side-channel capacities, can be solved avoiding
exponential path enumeration entirely. Finally, we showed the
practicality of our methods by using them to compute the timing
side-channel capacities of two commonly-used programs with very large
numbers of paths. 

% future work goes here... extensions to allow adversary to specify a low-security input, or be adaptive? (compare to KB, which cannot handle many paths but does allow adaptive adversaries)

\section*{Acknowledgements}

Daniel wishes to thank Jon Kotker, Rohit Sinha, and Zach Wasson for providing assistance with some of the tools used in our experiments, and Garvit Juniwal for a useful discussion of Algorithm \ref{algorithm-sums} and Lemma \ref{lemma-sumset-size}. The authors also thank the anonymous reviewers for their helpful comments and suggestions. This work was supported in part by the TerraSwarm Research Center, one of six centers supported by the STARnet phase of the Focus Center Research Program (FCRP), a Semiconductor Research Corporation program sponsored by MARCO and DARPA.

%=====================================================================
%\bibliographystyle{splncs}
\bibliographystyle{abbrv}
\bibliography{main}

\appendix

\section{Proofs}

\subsection{Lemmas from Sec. \ref{sec:algos-theory}}

\newtheorem*{lemma:AlgorithmSupports}{Lemma \ref{lemma-algorithm-supports}}
\begin{lemma:AlgorithmSupports}
Algorithm \ref{algorithm-supports} computes the supports $\support{b}$ correctly, and given an SMT oracle runs in time polynomial in $\abs{\program}$ and $\abs{\inputs}$.
\end{lemma:AlgorithmSupports}
\begin{proof}
Correctness is obvious. Computation of the CFG can clearly be done in
time linear in $\abs{\program}$, and likewise for finding nested
conditionals (say by doing a DFS and keeping track of the nesting
level). Generating SMT representations of the branch conditions and
doing dependency analyses on them can be done in time polynomial in
$\abs{\program}$. In the worst-case scenario where every input bit
appears in every branch condition, checking all the cofactor
equivalences requires $\abs{\branchpoints} \abs{\inputs}$ calls to the
SMT solver (generating the SMT query for a single cofactor equivalence
obviously takes time linear in $\abs{\program}$). So the algorithm
runs in time polynomial in $\abs{\program}$ and $\abs{\inputs}$.
\end{proof}

\newtheorem*{lemma:PathProb}{Lemma \ref{lemma-path-prob}}
\begin{lemma:PathProb}
Picking $i \in \{0,1\}^{\inputs}$ uniformly at random, for any path
$P$, the probability that the path corresponding to input $i$ is $P$
is given by
\begin{equation*}
\Pr \left[ \path{i} = P \right] = \left[ 2^{\abs{\nonesupport}} \left( \prod_{b \in \branches{P}} \truecount{b} \right) \left( \prod_{b \in \branchpoints \setminus \branches{P}} \left( 2^{\abs{\support{b}}} - \truecount{b} \right) \right) \right] / 2^{\abs{\inputs}} \enspace.
\end{equation*}
\end{lemma:PathProb}
\begin{proof}
We show that the product in square brackets is the number of $i \in
\{0,1\}^{\inputs}$ such that $\path{i} = P$. Since the conditionals of
$\program$ are unnested, $\path{i} = P$ iff $i$ makes $\condition{b}$
true for exactly those $b \in \branches{p}$. To specify $i$ we must
give its values on $\nonesupport$, on the sets $\support{b}$ for $b
\in \branches{p}$, and on the sets $\support{b}$ for $b \in
\branchpoints \setminus \branches{p}$. On $\nonesupport$ the bits may
have any value, since they do not affect any of the branch conditions,
giving the first factor of the product. On $\support{b}$ for $b \in
\branches{p}$ the bits must be set to make $\condition{b}$ true, and
by definition there are $\truecount{b}$ ways of doing this, giving the
second factor. Finally, on $\support{b}$ for $b \in \branchpoints
\setminus \branches{p}$ the bits must be set to make $\condition{b}$
false, and there are $2^{\abs{\support{b}}} - \truecount{b}$ ways of
doing this, giving the third factor.
\end{proof}

\newtheorem*{lemma:SubmultisetSums}{Lemma \ref{lemma-submultiset-sums}}
\begin{lemma:SubmultisetSums}
Let $D$ be the multiset of differences $\wt{\basispath{b}} -
\wt{\nonepath}$ for $b \in \branchpoints$. Then the possible values of
$\wt{i}$ over all inputs $i \in \{0,1\}^{\inputs}$ are the possible values of
$\wt{\nonepath} + D^+$, where $D^+$ is the set of sums of
submultisets of $D$. 
\end{lemma:SubmultisetSums}
\begin{proof}
By Lemma \ref{lemma-path-prob}, unless there is a ``fake'' branch
point whose condition is identically true or false ($\abs{\support{b}}
= 0$), every path has nonzero probability and is therefore feasible (we can detect and eliminate fake branch points when we
compute the condition supports).
%So the possible values of $\wt{x}$
%are just the possible values of $\wt{P}$ over all paths $P$. 

Now for any path $P$, by Equation \ref{equation-path-rep} and
linearity of the weight function we have 
\begin{align}
\label{equation-weight}
\wt{P} &= \left( \sum_{c \in \branches{P}} \wt{\basispath{c}} \right) - \left( \abs{\branches{P}} - 1 \right) \wt{\nonepath} \nonumber \\
&= \wt{\nonepath} + \sum_{c \in \branches{P}} \left( \wt{\basispath{c}} - \wt{\nonepath} \right) \enspace.
\end{align}
Since the conditionals of $\program$ are unnested, for every $B'
\subseteq \branchpoints$ there is some path $P$ such that $B' =
\branches{P}$. Thus by Equation \ref{equation-weight} the possible
values of $\wt{P}$ are all sums of elements of $D$ shifted by
$\wt{\nonepath}$.
\end{proof}

\newtheorem*{lemma:SumsetSize}{Lemma \ref{lemma-sumset-size}}
\begin{lemma:SumsetSize}
If $V = \max D - \min D$, then $\abs{D^+} = O(V \abs{D}^2)$. 
\end{lemma:SumsetSize}
\begin{proof}
We may assume $D$ has at least two elements --- list these in
increasing order as $d_1, \dots, d_n$. Letting $\tilde{D}$ be the
multiset of the values $\tilde{d}_i = d_i - d_1$, we have $0 =
\tilde{d}_1 \le \dots \le \tilde{d}_n = V$. Now for any $s \in
\tilde{D}^+$ we have $0 \le s \le \sum_i \tilde{d}_i \le (n-1)V$, and
so $\abs{\tilde{D}^+} \le (n-1)V + 1$. Finally, observe that for any
$0 \le k \le n$ and indices $1 \le i_1 < \dots < i_k \le n$, we have
$d_{i_1} + \dots + d_{i_k} = d'_{i_1} + \dots + d'_{i_k} + k
d_1$. Therefore we have $\abs{D^+} \le \abs{\tilde{D}^+} (n+1) \le
(n+1) \left[ (n-1)V + 1 \right] = O(V \abs{D}^2)$.
\end{proof}

\subsection{Analysis of Algorithm \ref{algorithm-sums}}
\label{section-sums-analysis}

Algorithm \ref{algorithm-sums} can easily be adapted to handle
arbitrary integers by removing any occurrences of 0 in $D$ and altering the inner loop so that when $x < 0$, we
enumerate $sums$ and build up $newSums$ from right to left instead of
from left to right. Since the unmodified algorithm is simpler to state
and slightly faster (having one less conditional in the outer loop),
we restrict our analysis to that case. Note however that if only the
\emph{size} of $D^+$ is needed and not its elements (as when computing
channel capacity), we can apply the following lemma: 
\begin{lemma}
If $D$ is a multiset of integers, and $T$ is $D$ with the absolute
value applied to all of its elements, then $\abs{D^+} = \abs{T^+}$. 
\end{lemma}
\begin{proof}
We show that if $D = R \cup \{ x \}$ and $\tilde{D} = R \cup \{ -x
\}$, then $D^+ = \tilde{D}^+ + x$. This suffices to prove the general
case, since flipping the signs on all negative elements of $D$ one by
one and applying the above result each time shows that $D^+$ is $T^+$
shifted by some constant. 

Take any $y \in D^+$. If $y$ can be written as a sum of elements of
$R$, then letting $y'$ be the same sum plus $-x$ we have $y' \in
\tilde{D}^+$ and thus $y = y' + x \in \tilde{D}^+ + x$. Otherwise, $y$
equals $x$ plus some sum of elements of $R$, and we let $y'$ be the
latter sum. Then $y' \in \tilde{D}^+$, and again $y = y' + x \in
\tilde{D}^+ + x$. So $D^+ \subseteq \tilde{D}^+ + x$. Conversely, take
any $y \in \tilde{D}^+ + x$, so that $y$ is $x$ plus a sum of elements
of $\tilde{D}$. If this sum contains $-x$, then $y$ is just equal to a
sum of elements of $R$, and so is in $D^+$. Otherwise, $y$ is $x$ plus
a sum of elements of $R$, and so again is in $D^+$. Therefore
$\tilde{D}^+ + x \subseteq D^+$, and so $D^+ = \tilde{D}^+ + x$. 
\end{proof}
So if we only need $\abs{D^+}$, as a preprocessing step we can take
the absolute value of all elements of $D$ to ensure they are
positive (removing 0), and then apply the unmodified Algorithm
\ref{algorithm-sums}. This was done in our experiments. Now we prove 

% this proof should be written more cleanly!
\begin{theorem}
\label{theorem-algorithm-sums}
Algorithm \ref{algorithm-sums} is correct, and has worst-case runtime $\Theta(\abs{D} \abs{D^+})$.
\end{theorem}
\begin{proof}
We prove that if $sums$ is a list of distinct nonnegative integers
sorted in increasing order, the body of the loop on line
\ref{outer-loop} results in $sums$ being updated to include all
integers of the form $s + x$ for $s$ in $sums$, still in increasing
order and with no duplicates. Since $sums$ is initially set to be the
list with the single element $0$ on line \ref{countsums-init}, it will
follow by induction that on line \ref{countsums-return} the list
$sums$ is $D^+$ sorted in increasing order. So the algorithm returns
$D^+$, and is correct. 

For notational convenience we will sometimes refer to $sums$ and
$newSums$ as sets. We need to show that at line
\ref{countsums-update}, $newSums$ lists the set $sums \cup (sums + x)$
in increasing order without duplicates. It is clear from lines
\ref{countsums-newsums-init}, \ref{countsums-oldsum}, and
\ref{countsums-newsum} that $newSums$ is contained in $sums \cup (sums
+ x)$, and from line \ref{countsums-newsum} that $sums + x$ is
contained in $newSums$. From the conditions on lines
\ref{countsums-scanloop} and \ref{countsums-duplicate}, we see that
while $i$ is less than $\mathsf{len}(sums)$, it is only incremented
when $sums[i]$ has been added to $newSums$, either by line
\ref{countsums-oldsum} or by line \ref{countsums-newsum} if $sums[i] =
z$. When $y$ is the last, and thus the unique largest, element of
$sums$, either $z = x + y$ is larger than every element of $sums$ or
equal to $y$ (since $x$ is nonnegative). In either case, the loop at
line \ref{countsums-scanloop} will repeat until $i =
\mathsf{len}(sums)$. Therefore, since $newSums$ starts with $sums[0]$
from line \ref{countsums-newsums-init}, at line \ref{countsums-update}
every element of $sums$ will be in $newSums$, and so $newSums$ lists
the set $sums \cup (sums + x)$. By the conditions on line
\ref{countsums-scanloop}, no value of $z$ is added to $newSums$ unless
all smaller values of $sums$ and $sums + x$ have already been added,
since $sums$ is in increasing order. Furthermore, if $z$ equals some
value in $sums$, say with index $i$, then the check on line
\ref{countsums-duplicate} ensures that $i$ is incremented so that the
value $z$ is only added once to $newSums$ (and since $z \ge x > 0$, we have $z > sums[0]$ and thus the value added to $newSums$ on line \ref{countsums-newsums-init} is not duplicated). Therefore at line
\ref{countsums-update}, $newSums$ is in increasing order and has no
duplicates, as desired. 

From our work above, we see that in every iteration of the loop on
line \ref{outer-loop} the variable $i$ is incremented until $i =
\mathsf{len}(sums)$ and no further. Therefore in every such iteration,
the loop on line \ref{countsums-scanloop} takes
$O(\mathsf{len}(sums))$ time in total for all of its iterations. Since
the body of the loop on line \ref{countsums-inner-loop} takes constant
time excluding the loop on line $\ref{countsums-scanloop}$, every
iteration of the loop on line \ref{outer-loop} takes
$O(\mathsf{len}(sums))$ time. In every iteration $\mathsf{len}(sums)$
is bounded above by $\abs{D^+}$, since $sums$ never gets shorter and
after the last iteration has length exactly $\abs{D^+}$. Since there
are exactly $\abs{D}$ iterations of the loop on line \ref{outer-loop},
the entire algorithm runs in $O(\abs{D} \abs{D^+})$ time. If $D$
consists of $n$ copies of 1, it is easy to see that $sums$ grows
linearly from length 1 to length $n+1$, so that the algorithm runs in
$\Omega(n^2) = \Omega(\abs{D} \abs{D^+})$ time. 
\end{proof}

\end{document}